\documentclass{article}
\usepackage{fullpage,amsmath,amssymb,amsfonts,amsthm}
\usepackage{tikz}
\usepackage{color}

\newtheorem{theorem}{Theorem}[section]
\newtheorem{lemma}[theorem]{Lemma}
\newtheorem{prop}[theorem]{Proposition}
\newtheorem{corollary}[theorem]{Corollary}

\newtheorem{question}[theorem]{Question}
\newtheorem{remark}[theorem]{Remark}

\theoremstyle{definition}
\newtheorem{definition}[theorem]{Definition}

\newcommand{\F}{\mathbb{F}}

\renewcommand{\S}{\mathcal{S}}

\renewcommand{\P}{\mathcal{P}}

\newcommand{\comment}[1]{}
\newcommand{\nts}[1]{}

\newcommand{\lb}{e^{\sqrt{n/2}}}

\title{A note on the minimum distance of quantum LDPC codes}
\author{Nicolas Delfosse\footnote{D\'epartement de Physique, Universit\'e de Sherbrooke, Sherbrooke, Qu\'ebec, J1K 2R1, Canada, \texttt{nicolas.delfosse@usherbrooke.ca}}, Zhentao Li\footnote{D\'epartement d'Informatique UMR CNRS 8548, \'Ecole Normale Sup\'erieure \texttt{zhentao.li@ens.fr}} and St\'ephan Thomass\'e\footnote{\'Ecole Normale Sup\'erieure de Lyon, LIP , \'Equipe MC2 {\tt stephan.thomasse@ens-lyon.fr}}}

\begin{document}

\maketitle

\begin{abstract}
We provide a new lower bound on the minimum distance of a family of quantum LDPC codes based on Cayley graphs proposed by MacKay, Mitchison and Shokrollahi \cite{MMS07}. Our bound is exponential, improving on the quadratic bound of Couvreur, Delfosse and Z\'emor \cite{CDZ13}. This result is obtained by examining a family of subsets of the hypercube which locally satisfy some parity conditions.
\end{abstract}

\section{Introduction}

%
%
%

A striking difference between classical and quantum computing is the unavoidable presence of perturbations when we manipulate a quantum system, which induces errors at every step of the computation. This makes essential the use of quantum error correcting codes. Their role is to avoid the accumulation of errors throughout the computation by rapidly identifying the errors which occur.

One of the most satisfying construction of classical error correcting codes capable of a rapid determination of the errors which corrupt the data is the family of Low Density Parity--Check codes (LDPC codes) \cite{Ga63}. It is therefore natural to investigate their quantum generalization. Moreover, Gottesman remarked recently that this family of codes can significantly reduce the overhead due to the use of error correcting codes during a quantum computation \cite{Go13}. Quantum LDPC codes may therefore become an essential building block for quantum computing.

Quantum LDPC codes have been proposed by MacKay, Mitchison, and MacFadden in \cite{MMM04}. One of the first difficulty which arises is that most of the families of quantum LDPC codes derived from classical constructions lead to a bounded minimum distance, see \cite{TZ09} and references therein. Such a distance is generally not sufficient and it induces a poor error-correction performance.

Only a rare number of constructions of quantum LDPC codes are equipped with an unbounded minimum distance. Most of them are inspired by Kitaev toric codes constructed from the a tiling of the torus \cite{Ki97} such as,  color codes which are based on 3-colored tilings of surfaces \cite{BM06:color},
hyperbolic codes which are defined from hyperbolic tilings \cite{FML02, Ze09},
or other constructions based on tilings of higher dimensional manifolds \cite{FML02, GL13}.
These constructions are based on tilings of surfaces or manifolds and their minimum distance depends on the homology of this tiling. The determination of the distance of these codes is thus based on homological properties and general bounds on the minimum distance can be derived from sophisticated homological inequalities \cite{Fe12, De13:tradeoffs}.

In this article, we study a construction of quantum LDPC codes based on Cayley graphs which has been proposed by MacKay, Mitchison and Shokrollahi \cite{MMM04} and has been studied in \cite{CDZ13}. This family does not rely on homological properties and thus the homological method cited earlier seems impossible to apply. We relate the determination of this minimum distance with a combinatorial problem in the hypercube. Then, using an idea of Gromov, we derive a lower bound on the minimum distance of these quantum codes which clearly improves the results of Couvreur, Delfosse and Z\'emor \cite{CDZ13}.

The remainder of this article is organized as follows. In Section~\ref{section:background}, we recall the definition of linear codes and a construction of quantum codes based on classical codes. Section~\ref{section:cayley_codes} introduces the quantum codes of MacKay, Mitchison and Shokrollahi \cite{MMS07}. In order to describe the minimum distance of these quantum codes based on Cayley graphs, we introduced two families of subsets of these graphs that we call borders and pseudo-borders in Section~\ref{section:cayley_border}. We are then interested in the size of pseudo-borders of Cayley graphs. In Section~\ref{section:hypercube_border}, we reduce this problem to the study of $t$-pseudo-borders of the hypercube, which are a local version of pseudo-borders. Theorem~\ref{theo:bound_pseudo-border}, proved in Section~\ref{section:main_proof}, establishes a lower bound on the size of $t$-pseudo-border. As a corollary, we derive a lower bound on the minimum distance of Cayley graphs quantum codes.

\section{Minimum distance of quantum codes}
\label{section:background}

A \emph{code of length $n$} is defined to be a subspace of $\F_2^n$. It contains $2^k$ elements, called \emph{codewords}, where $k$ is the dimension of the code. The \emph{minimum distance} $d$ of a code is the minimum Hamming distance between two codewords. By linearity, it is also the minimum Hamming weight of a non-zero codeword. This parameter plays an important role in the error correction capability of the code. Indeed, assume we start with a codeword $c$ and that $t$ of its bits are flipped. Denote by $c'$ the resulting vector. If $t$ is smaller than $(d-1)/2$ then we can recover $c$ by looking for the closest codeword of $c'$. Therefore, we can theoretically correct up to $(d-1)/2$ bit-flip errors. The parameters of a code are denoted $[n, k, d]$.

Every code can be defined as the kernel of a binary matrix $H$. This matrix is called a \emph{parity--check matrix} of the code. Alternatively, a code can be given as the space generated by the rows of a matrix. This matrix is called a \emph{generator matrix} of the code. For instance, the following parity--check matrix defines a code of parameter $[7, 4, 3]$.
\begin{equation} \label{eqn:H}
H = 
\begin{pmatrix}
1 & 0 & 1 & 0 & 1 & 0 & 1\\
0 & 1 & 1 & 0 & 0 & 1 & 1\\
0 & 0 & 0 & 1 & 1 & 1 & 1\\
\end{pmatrix}
\cdot
\end{equation}
The code which admits $H$ as a generator matrix has parameters $[7, 3, 4]$.

The space $\F_2^n$ is equipped with the inner product $(x, y) = \sum_{i=1}^n x_i y_i$, where $x=(x_1, x_2,\dots, x_n)$ and $y=(y_1, y_2, \dots, y_n)$. The orthogonal of a code $C$ of length $n$ is called the \emph{dual code} of $C$. A code is said \emph{self-orthogonal} it it is included in its dual. For example, we can easily check that two rows of the matrix $H$ given in Eq.(\ref{eqn:H}) are orthogonal which means the code generated by the rows of $H$ is self-orthogonal.

Quantum information theory studies the generalization of error correcting codes to the protection of an information written in a quantum mechanical system. By analogy with the classical setting, a quantum error correcting code is defined as the embedding of $K$ qubits into $N$ qubits. The CSS construction allows us to define a quantum code from a classical self-orthogonal code \cite{CS96, St96}. As in the classical setting, the minimum distance $D$ of a quantum code is an important parameter which measures the performance of the quantum code. The following proposition gives a combinatorial description of the parameters of these quantum codes.
\begin{prop} \label{prop:quantum_code}
Let $C$ be a classical code of parameters $[n, k, d]$.
If $C$ is self-orthogonal, then we can associate with $C$ a quantum code of parameters $[[N, K, D]]$, where $N=n$, $K= \dim C^\perp/C = n-2k$ and when $K\neq 0$, $D$ is the minimum weight of codeword of $C^\perp$ which is not in $C$:
$$
D = \min \{ w(x) \ | \ x \in C^\perp \backslash C \}.
$$
\end{prop}
Throughout this article, we only consider this combinatorial definition of the parameters of quantum codes. A complete description of quantum error correcting codes, starting from the postulates of quantum mechanics, can be found for example in \cite{NC00}.

The minimum distance of the classical code $C^\perp$ is $d^\perp = \min\{ w(x) \ | \ x \in C^\perp \backslash \{0\} \}$.  If there exists a vector $x \in C^\perp \backslash \{0\}$ of minimum weight which is not in $C$, then the quantum minimum distance is $D=d^\perp$. In that case, the computation of the minimum distance corresponds to the computation of the minimum distance of the classical code $C^\perp$.

When, $D$ is strictly larger than the classical minimum distance $d^\perp$, the quantum code is said \emph{degenerated}. Then, we do not consider the codewords of $C$ in the computation of $D$. This essential feature can improve the performance of the quantum code but it also makes the determination of the minimum distance strikingly more difficult than in the classical setting. In the present work, we obtain a lower bound on the minimum distance of a family of degenerate quantum codes.

\section{A family of quantum codes based on Cayley graphs}
\label{section:cayley_codes}

We consider a family of quantum codes constructed from Cayley graphs, which we now define.

\begin{definition} \label{defi:cayley}
Let $G$ be a group and $S$ be a set of elements of $G$ such that $s \in S$ implies $s^{-1} \in S$. The  \emph{Cayley graph} $\Gamma(G, S)$ is the graph with vertex set $G$ such that two vertices are adjacent if they differ by an element in $S$.
\end{definition}

In our case, the group $G$ is always $\F_2^r$ and $S = \{c_1, c_2, \dots, c_n\}$ is a generating set of $\F_2^r$. Thus $\Gamma(G, S)$ has $2^r$ vertices and it is a regular graph of degree $n$. This graph is connected since $S$ is a generating set.
To simplify notation, we assume that the vectors $c_i$ are the $n$ columns of a matrix $H \in M_{r, n}(\F_2)$. We denote by $G(H)$ this graph and we denote by $A(H)$ the adjacency matrix of this graph.

For instance, the Cayley graph $G(I_n)$, associated with the identity matrix of size $n$, is the hypercube of dimension $n$. Indeed, its vertex set is $\F_2^n$ and two vertices $x$ and $y$ are adjacent if and only the vectors $x$ and $y$ differ in exactly one component.

The following proposition proves that we can associate a quantum code with these graphs \cite{CDZ13}.
\begin{prop} \label{prop:cayley_codes}
Let $H \in M_{r, n}(\F_2)$ be a binary matrix. If $n$ is an even integer, then the adjacency matrix $A(H)$ of the graph $G(H)$ is the generating matrix of a classical self-orthogonal code. We denote by $C(H)$ this self-orthogonal code and by $Q(H)$ the corresponding quantum code.
\end{prop}
We want to determine the minimum distance of these quantum codes $Q(H)$.

Recall that a family of quantum codes associated with a family of self-orthogonal codes $(C_i)$ defines quantum LDPC codes if $C_i$ admits a parity--check matrix $H_i$ which is sparse. In our case the generating matrix of the code $C(H)$ is the adjacency matrix $A(H)$, which is typically sparse since each row is a vector of length $2^r$ and weight $n$. When $n$ and $r$ are proportional, the each row of $A(H)$ has weight in $O(log N)$, where $N=2^r$ is the number of columns of $A(H)$. Some authors consider LDPC codes defined by a parity-check matrix with bounded row weight. However, an unbouded row weight is needed for example to achieve the capacity of the quantum erasure channel \cite{DZ13}.

It turns out that regardless of the choice of generators, the resulting graph $G(H)$ is always \emph{locally isomorphic} to the hypercube \cite{CDZ13}. That is, the set of vertices within some distance $t$, depending on $H$, of a vertex in $G(H)$, is isomorphic to the subgraph of the hypercube induced by all vertices within distance $t$ of a vertex in the hypercube.

We close this subsection with a sketch of proof of this local isomorphism. By \emph{ball} of radius $t$ centered at a vertex $v$ in a graph $G$, we mean the subgraph of $G$ induced by all vertices at distance at most $t$ from $v$.

The radius of the isomorphism depends on the shortest length $d$ of a relation $\sum_{i=1}^d c_i = 0$ between columns of $H$, which is, by definition, the minimum distance of the code of parity-check matrix $H$.
\begin{prop} \label{prop:local_isomorphism}
Let $H \in M_{r, n}(\F_2)$ and let $d$ be the minimum distance of the code of parity--check matrix $H$.
Then, there if a graph isomorphism between any ball of radius $(d-1)/2$ of $G(H)$ and any ball of same radius in $G(I_n)$ where $I_n$ is the identity matrix of size $n$.
\end{prop}

\begin{proof}[Proof sketch]
The generating set $S$ is the set of columns of $H$. The neigbourhood of a vertex $x$ is $\{x+c_1| c_1 \in S\}$ and the second neighbourhood is $\{x+c_1+c_2| c_1 \neq c_2 \in S\}$.
If there is no relation between set of 4 different generators (i.e. there is no relation where the sum of four generators is $0$) then except for $x+c_1+c_2$ and $x+c_2+c_1$ begin equal, all these vertices are distinct. So by mapping $x$ to $0$, $c_i$ to $e_i$ and $c_i+c_j$ to $e_i+e_j$, we see that vertices at distance at most two from $x$ is isomorphic to the hypercube.

In the general case, we can also map each generator of $G(H)$ to a generator of $G(I_n)$.
\end{proof}

This result gives more information about the local structure of the graph when we start with a parity--check matrix $H$ defining a code of large distance $d$. For instance, when $H \in M_{r, n}(\F_2)$ is a random matrix chosen uniformly among the matrices of maximal rank, the distance $d$ is linear in $n$. We can also choose $H$ as the parity-check matrix of a known code equipped with a large minimum distance.

Our result will be proved by only looking at vertices within distance $(d-1)/2$ of some central vertex and hence we may assume that we are in the hypercube of dimension $n$.

\section{Borders and pseudo-borders of Cayley graphs}
\label{section:cayley_border}

The aim of this section is to provide a graphical description of the minimum distance of the quantum codes $Q(H)$ introduced in the previous section. This quantum code is associated with the classical self-orthogonal code $C(H)$ generated by the rows of the matrix $A(H)$. By Proposition~\ref{prop:quantum_code}, the minimum distance $D$ of this quantum code is given by
$$
D = \min \{ w(x) \ | \ x \in C(H)^\perp \backslash C(H) \}.
$$
By definition, the codes $C(H)$ and its dual $C(H)^\perp$ are subspaces of $\F_2^N$ where $N = 2^r$ is the size of the matrix $A(H)$. Since the columns of $A(H)$ are indexed by the $N$ vertices of the graph $G(H)$, a vector $x \in \F_2^N$ can be regarded as the indicator vector of a subset of the vertex set of $G(H)$. We can then replace the two conditions $x \in C(H)^\perp$ and $x \notin C(H)$ by conditions on the set of vertices corresponding to $x$. In order to describe the vectors $x$ of $C(H)$ and $C(H)^\perp$, we introduced two families of subsets of the vertex set of $G(H)$: the borders and the pseudo-borders.

The neighbourhood $N(v)$ of a vertex $v$ of the graph $G(H)$ is the set of vertices incident to $v$.
\begin{definition} \label{defi:border}
Let $\S$ be a subset of the vertex set of $G(H)$. The \emph{border} $B(\S)$ of $\S$ in the graph $G(H)$ is the set of vertices of $G(H)$ which belong to an odd number of neighbourhoods $N(v)$ for $v \in \S$.
\end{definition}
Equivalently, the border of a subset $\S$ is the symmetric difference of all the neigborhoods $N(v)$ for $v \in \S$.

\begin{definition} \label{defi:pseudo-border2}
A \emph{pseudo-border} in the graph $G(H)$ is a family $\P$ of vertices of $G(H)$ such that the cardinality of $N(v)\cap \P$ is even for every vertex $v$ of $G(H)$.
\end{definition}

These borders and pseudo-borders correspond to the vectors of the classical code $C(H)$ and its dual.
\begin{prop} \label{prop:correspondance}
Let $x$ be a vector of $\F_2^N$ where $N = 2^r$ is the number of vertices of $G(H)$. Then $x$ is the indicator vector of a subset $\S_x$ of the vertex set of $G(H)$. Moreover, we have
\begin{itemize}
\item $x \in C(H)$ if and only if $\S_x$ is border,
\item $x \in C(H)^\perp$ if and only if $\S_x$ is pseudo-border.
\end{itemize}
\end{prop}

\begin{proof}
By definition of the code $C(H)$, we have $x \in C(H)$ if and only if it is a linear combination of the rows of $A(H)$: $x = \sum_{i=1}^N \lambda_i r_i$.
Since $A(H)$ is the adjacency matrix of the graph $G(H)$, each row $r_i$ is the indicator vector of the neighbourhood $N(v_i)$ of a vertex $v_i$ of $G(H)$.
Then, remark that the sum $x+y$ of two vectors of $\F_2^N$ is the indicator vector of the symmetric difference of the sets $S_x$ and $S_y$, i.e. $S_{x+y} = S_x \Delta S_y$.
This proves that the vector $x = \sum_i \lambda_i r_i$ correspond to the symmetric difference of the neighbourhoods $N(v_i)$ such that $\lambda_i = 1$. The vector $x$ is indeed a border. The proof of the inverse implication is similar.

To prove the second property, it suffices to remark that $x \in C(H)^\perp$ if and only if $x$ is orthogonal to every row of $A(H)$. Then, the orthogonality between $x$ and the rows $r_i$ is equivalent to the fact that the set $S_x$ contains an even number of vertices of the set $N(v_i)$. Recall that $r_i$ is the indicator vector of the set $N(v_i)$. This proves the proposition.
\end{proof}

This proposition, combined with Proposition~\ref{prop:cayley_codes}, shows that, when $n$ is even, every border is a pseudo-border. In some special cases, every pseudo-border is a border. However, this is generally not true. When the graph $G(H)$ contains pseudo-borders which are not borders, the minimum distance of the quantum code associated with $H$ is equal to
$$
D = \min\{ |\S| \ | \ \S \text{ is a pseudo-border which is not a border } \}.
$$
Since the graph $G(H)$ is locally isomorphic to the hypercube of dimension $n$, (Proposition~\ref{prop:local_isomorphism}), it is natural to investigate the borders and pseudo-borders of the hypercube.

\section{Borders and pseudo-borders of the hypercube}
\label{section:hypercube_border}

In this section, we focus on the hypercube and we introduce a local version of the pseudo-borders which is preserved by the local isomorphism of Proposition~\ref{prop:local_isomorphism}.

Earlier, the hypercube appeared as the Cayley graph $G(I_n)$. In what follows, we use an alternative definition of this graph. Let $n$ be some integer. We denote by $[n]$ the set $\{1,2,\ldots, n\}$. This set can be regarded as the vertex set of a hypercube as represented in Figure~\ref{fig:hypercube}. The following definition is inspired by this graphical structure. Let $S$ be a subset of $[n]$, we denote by $N(S)$ the \emph{neighbourhood} centered at $S$, i.e. the family of all subsets of $[n]$ that differ by one element from $S$.

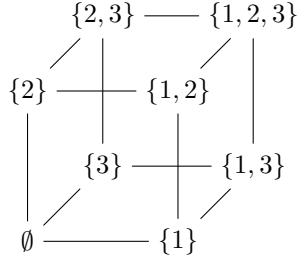
\begin{figure}[htbp]
\begin{center}
\begin{tikzpicture}
    \node (o) at ( 0,-2) {$\emptyset$};
    \node (x) at ( 2,-2) {$\{1\}$};
    \node (y) at ( 1,-1) {$\{3\}$};
    \node (z) at ( 0, 0) {$\{2\}$}; 
    \node (xy) at ( 3,-1) {$\{1,3\}$};
    \node (xz) at ( 2,0) {$\{1,2\}$};
    \node (yz) at ( 1,1) {$\{2,3\}$};
    \node (xyz) at ( 3,1) {$\{1,2,3\}$};
    
    \begin{scope}
       \draw (o) -- (x);
       \draw (o) -- (y); 
       \draw (o) -- (z);
       \draw (x) -- (xy);
       \draw (x) -- (xz);
       \draw (y) -- (yz);
       \draw (y) -- (xy);
       \draw (z) -- (xz);
       \draw (z) -- (yz);
       \draw (xy) -- (xyz);
       \draw (xz) -- (xyz);
       \draw (yz) -- (xyz);

    \end{scope}  
\end{tikzpicture}
\caption{The hypercube based on the set $[3]$}
\label{fig:hypercube}
\end{center}
\end{figure}

Before starting with the definitions, let us recall a simple lower bound on the cardinality of a pseudo-border $\S$ which is not a border in $G(H)$. Assume that the parameter $d$ is larger than 7, so that every ball of radius $(d-1)/2=3$ of the graph $G(H)$ is isomorphic to a ball of the hypercube of dimension $n$ by Proposition~\ref{prop:local_isomorphism}. The pseudo-border $\S$ is not empty, otherwise it is a border. Therefore, it contains  a vertex $u$ of $G(H)$. We use the following arguments:
\begin{itemize}
\item $u \in \S$,
\item The cardinality of $\S \cap N(v)$ is even for every neighbour $v$ of the vertex $u$,
\item The ball $B(u, 3)$ of the graph $G(H)$ is isomorphic to a ball of the hypercube of dimension $n$.
\end{itemize}
Then, each of the $n$ neighbourhoods $N(v)$ centered at $v \in N(u)$, contains the vertex $u$, thus it must contain at least another vertex of $\S$. This provides $n$ other vertices of the set $\S$. Since these vertices can appear in at most two different sets $N(v)$, the set $\S$ contains at least $1+n/2$ distinct vertices.

In order to extent this argument to a larger ball of the hypercube, we introduce $t$-pseudo-borders.
\begin{definition} \label{defi:t-pseudo-border}
Let $t$ and $n$ be two positive integers such that $t<n$. A $t$-pseudo-border $\S$ of the hypercube $[n]$ is a family of subset of $[n]$ such that
\begin{itemize}
\item $\emptyset \in \S$,
\item The cardinality of $\S \cap N(S)$ is even for every $S \subset [n]$ of size $|S| \leq t-1$,
\item $\S$ is included in the ball of radius $t$ centered in $\emptyset$ of the hypercube $[n]$.
\end{itemize}
\end{definition}

In other words, a $t$-pseudo-border is a subset of vertices of a ball of the hypercube satisfying the conditions of the definition of a pseudo-border in this ball. Starting from a pseudo-border of a Cayley graph $G(H)$ and applying the local isomorphism of Proposition~\ref{prop:local_isomorphism}, we obtain a $t$-pseudo-border of the hypercube. We are interested in a lower bound on the size of $t$-pseudo-borders.
Thus, we aim to answer the following refinement of the question of the determination of the minimum distance of the quantum codes $Q(H)$.

\begin{question} \label{q:hypercube}
What is the minimum cardinality of a $t$-pseudo-border of the hypercube $[n]$?
\end{question}

The results of \cite{CDZ13} provide a polynomial lower bound in $O(tn^2)$. To our knowledge, this is the best known lower bound on the cardinality of the $t$-pseudo-border of the hypercube. 
Our main result is an exponential lower bound.

\begin{theorem} \label{theo:bound_pseudo-border}
The minimum cardinality of a $t$-pseudo-border of the hypercube $[n]$ is at least 
$$
\sum_{i=0}^{i \leq M}\frac{(n/2)^{i/2}}{i!},
$$
where $M = \min\{t-1, \sqrt{n/2}\}$.
When $t$ is larger than $\sqrt{n/2}$, this lower bound is at least $\lb$.
\end{theorem}

This Theorem is proved in Section~\ref{section:main_proof}. As an application, we obtain a lower bound on the minimum distance of the Cayley graph quantum codes $Q(H)$.

\begin{corollary} \label{cor:bound_distance}
Let $H \in M_{r, n}(\F_2)$, with $n$ an even integer, and let $d$ be the minimum distance of the code of parity--check matrix $H$. The quantum code $Q(H)$ encodes $K$ qubits into $N=2^r$ qubits. If $K\neq 0$, then the minimum distance $D$ of $Q(H)$ is at least
$$
D \geq \sum_{i=0}^{i \leq M}\frac{(n/2)^{i/2}}{i!},
$$
where $M = \min\{(d-3)/2, \sqrt{n/2}\}$.
When $d$ is larger than $\sqrt{n/2}$, this lower bound is at least $\lb$.
\end{corollary}

\begin{proof}
We want to bound the minimum distance $D$ of $Q(H)$, which, by Proposition~\ref{prop:quantum_code}, is the minimum weight of a vector $x$ of $C(H)^\perp \backslash C(H)$. Such a vector $x$ corresponds to a subset $\S_x$ of vertices of the graph $G(H)$ which is a pseudo-border and which is not a border by Proposition~\ref{prop:correspondance}. By this bijection $x\rightarrow \S_x$, the weight of $x$ corresponds to the cardinality of the set $\S_x$. Therefore, $D$ is the minimum cardinality of a pseudo-border of $G(H)$ which is not a border.

First, let us prove that such a pseudo-border exists. By Proposition~\ref{prop:quantum_code}, the number of encoded qubits $K$, which is assumed to be positive, is the dimension of the quotient space $K = \dim C(H)^\perp/C(H)$. We know, from Proposition~\ref{prop:correspondance}, that $C(H)^\perp$ and $C(H)$ are in one-to-one correspondence with the sets of pseudo-borders and borders respectively. Therefore the positivity of $K$ implies the existence of a pseudo-border $\S$ of $G(H)$ which is not a border.

Since $\S$ is not a border it is not empty. Let $u$ be a vertex of $\S$ in $G(H)$. Applying the local isomorphism of Proposition~\ref{prop:local_isomorphism}, we can map the ball of radius $(d-1)/2$ centered at $u$ of $G(H)$ to the ball of same radius centered at $\emptyset$ of the hypercube $[n]$. By this transformation, the restriction of $\S$ to the ball $B(u, (d-1)/2)$ is sent onto a $(d-1)/2$-pseudo-border $\S_u$ of the hypercube $[n]$.
This graph isomorphism cannot increase the size of $\S$ thus $D$ is lower bounded by the minimum size of a $(d-1)/2$-pseudo-border of $[n]$:
$$
D \geq |\S| \geq |\tilde \S| \geq \sum_{i=0}^{i \leq M}\frac{(n/2)^{i/2}}{i!},
$$
where $M = \min\{(d-3)/2, \sqrt{n/2}\}$.
The last inequality is derived from Theorem~\ref{theo:bound_pseudo-border}.
\end{proof}


\section{Bound on the size of local pseudo-borders of the hypercube}
\label{section:main_proof}

This section is devoted to the proof of Theorem~\ref{theo:bound_pseudo-border}. So our goal is to derive a lower bound on the size of $t$-pseudo-borders of the hypercube. Here it is more convenient to use the language of sets. We work in the hypercube $[n]$ defined in Section~\ref{section:hypercube_border}, whose vertices correspond to subsets of $[n]=\{1,2,\ldots, n\}$.

Our goal is to obtain a lower bound on the minimum cardinality of a $t$-pseudo-border. We consider a $t$-pseudo-border $\S$ of minimum size. In order to exploit the minimality of this set, we introduce the following operation.
\begin{definition} \label{defi:flipping}
Let $\S$ be a family of subsets of $[n]$ and let $S$ be a subset of $[n]$. By \emph{flipping} $\S$ along $S$, we mean to swap elements and non-elements of $\S$ in the border $N(S)$.
\end{definition}
Put differently, flipping $\S$ along $S$ gives the symmetric difference $\S \Delta N(S)$.
The following lemma is a key ingredient of the proof of the main result. It proves that flipping a minimum $t$-pseudo-border cannot decrease its size.
\begin{lemma} \label{lemma:flip-min}
Let $\S$ be a $t$-pseudo-border of $[n]$ of minimum cardinality. Then, flipping $\S$ along any number of subsets $S$ such that $2 \leq |S| \leq t-1$ leads to a family of greater or equal cardinality.
\end{lemma}

\begin{proof}
Since $\S$ is minimal among $t$-pseudo-borders, it suffices to show the symmetric difference of any $t$-pseudo-border $\S'$ with one neighbourhood $N(S)$ with $2 \leq |S| \leq t-1$ is still a $t$-pseudo-border. Since $|S| \geq 2$, $\emptyset$ is unaffected by this flip and remains in $\S'\Delta N(S)$. Since $|S| \le t-1$, elements outside the ball of radius $t$ centered at $\emptyset$ are unaffected by the flip and so $\S'\Delta N(S)$ is contained in this ball. It remains to see that $|(\S' \Delta N(S)) \cap N(T)|$ is even for every subset $T$ of $[n]$ of size $|T| \leq t-1$. We have $(\S' \Delta N(S)) \cap N(T) = (\S' \cap N(T)) \Delta (N(S) \cap N(T))$ and we know that $|A \Delta B| = |A| + |B| - 2|A \cap B| \equiv |A| + |B| \pmod 2$. Here $|\S' \cap N(T)|$ is even because $\S$ is a $t$-pseudo-border and $|N(S) \cap N(T)|$ is even since it is either $n$ (if $S=T$), 2 (if $S = T \cup \{x,y\}$, $T = S \cup \{x,y\}$ or $S = T-\{x\} \cup \{y\}$ for some $x,y \in [n]$) or 0. This proves that $\S' \Delta N(S)$ is a $t$-pseudo-border.
\end{proof}

We are therefore interested in minimal $t$-pseudo-borders, i.e. $t$-pseudo-borders of minimum cardinality.

\subsection{Lower bounds for 2-subsets and 4-subsets}

In Section~\ref{section:hypercube_border}, we showed a first lower bound on the size of pseudo-borders. The basic idea is to use the subset $\emptyset$ of $\S$ and the neighbourhood $N(S)$ centered in sets $S$ of size 1, to derive the existence of other subsets of $\S$ at distance $2$ from $\emptyset$. It is reasonable to expect that these subsets of $\S$ at distance 2 from $\emptyset$ will imply the existence of subsets of $\S$ at distance 4 from $\emptyset$ and so on. However, as the distance to $\emptyset$ increases, the local structure of the graph becomes more and more complicated and this problem is made more complex.

In this section, as an example, we give a lower bound on the number of 4-subsets of a minimal $t$-pseudo-borders. The same tools are then used in Section \ref{sect:generallb} to bound the number of $k$-subsets  of a minimal $t$-pseudo-border.

\begin{definition}
  A \emph{$k$-set} of a set $S$ is a subset of $S$ of size $k$.
\end{definition}

These $k$-sets are use to decompose the hypercube into layers. The set of $k$-subsets corresponds to the vertices at distance $k$ to $\emptyset$. The $k$-sets of the hypercube $[4]$ are represented in Figure~\ref{fig:k-sets}.

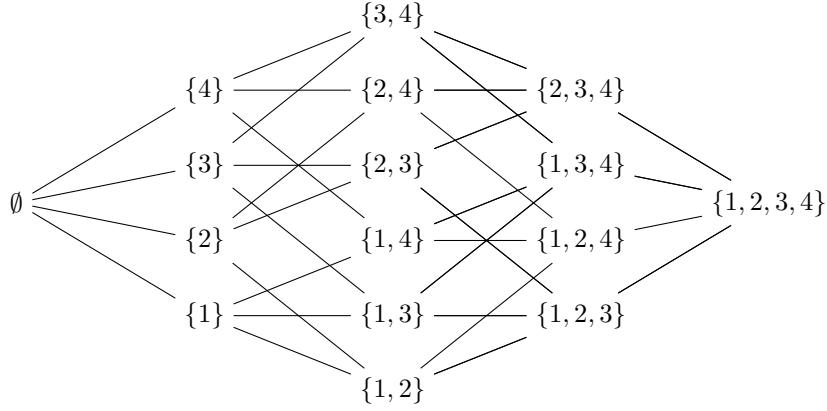
\begin{figure}[htbp]
\begin{center}
\begin{tikzpicture}
    \node (o) at ( -5,0) {$\emptyset$};
    \node (1) at ( -2.5,-1.5) {$\{1\}$};
    \node (2) at ( -2.5,-.5) {$\{2\}$};
    \node (3) at ( -2.5, .5) {$\{3\}$}; 
    \node (4) at ( -2.5, 1.5) {$\{4\}$};  
    \node (12) at ( 0,-2.5) {$\{1,2\}$};
    \node (13) at ( 0,-1.5) {$\{1,3\}$};
    \node (14) at ( 0,-.5) {$\{1,4\}$};    
    \node (23) at ( 0,.5) {$\{2,3\}$};
    \node (24) at ( 0,1.5) {$\{2,4\}$};
    \node (34) at ( 0,2.5) {$\{3,4\}$};    
    \node (123) at ( 2.5,-1.5) {$\{1,2,3\}$};
    \node (124) at ( 2.5,-.5) {$\{1,2,4\}$};
    \node (134) at ( 2.5, .5) {$\{1,3,4\}$}; 
    \node (234) at ( 2.5, 1.5) {$\{2,3,4\}$};
    \node (1234) at (5, 0) {$\{1,2,3,4\}$}; 
       
    \begin{scope}
    \foreach \x in {1, ..., 4}
       \draw (o) -- (\x);

    \foreach \x in {1, ..., 3}
    	\foreach \y in {\x, ..., 3}
    	{
    		\pgfmathtruncatemacro{\y}{\y+1}
    		\draw (\x\y)--(\x);
    		\draw (\x\y)--(\y);
    	}
    	
    \foreach \x in {1, ..., 2}
    	\foreach \y in {\x, ..., 2}
    		\foreach \z in {\y, ..., 3}
    	{
    		\pgfmathtruncatemacro{\y}{\y+1}
    		\pgfmathtruncatemacro{\z}{\z+1}
    		\draw (\x\y\z)--(\x\y);
    		\draw (\x\y\z)--(\x\z);
    		\draw (\x\y\z)--(\y\z);
    		\draw (\x\y\z)--(1234);
    	}

     \end{scope}
\end{tikzpicture}
\caption{The $k$-sets of the hypercube $[4]$.}
\label{fig:k-sets}
\end{center}
\end{figure}

\begin{definition}
  An \emph{odd $k$-set} $S$ of $[n]$ with respect to $\S$ is a $k$-set (not necessarily in $\S$) with an odd number of $k-1$-subsets of $\S$ contained in $S$.
\end{definition}

Stated differently, odd $k$-sets are the sets $S$ such that the constraint $|\S \cap N(S)|$ is even is not satisfied when we restrict $\S$ to the ball of radius $k-1$ centered in $\emptyset$. Therefore, they can be used to deduce the existence of $(k+1)$-subsets of $\S$ as proved in the following lemma.

\begin{lemma}\label{oddsettosets}
Let $k < t-1$.
If there are $s_k$ odd $k$-sets with respect to a minimal $t$-pseudo-border $\S$ then there are at least $\frac{s_k}{k+1}$ sets of size $k+1$ in $\S$.
\end{lemma}

\begin{proof}
  The ball centered at each of these odd $k$-sets contains an even number of elements of $\S$ and therefore contains at least one element of $\S$ of size $k+1$.

  On the other hand, each $k+1$-sets contain $k+1$ $k$-sets and therefore at most $k+1$ odd $k$-sets of $\S$. So we need at least $\frac{s_k}{k+1}$ such sets to satisfy the parity condition for all odd $k$-sets.
\end{proof}

To complete our proof, we now need to lower bound the number of odd $k+1$-sets in terms of the number of $k$-sets. Our proof is inspired by a result of Gromov, independently proven by Linial, Meshulam \cite{linial2006homological}, and Wallach \cite{meshulam2009homological} (and maybe by others). (See also Lemma 3 of \cite{matousek})

\begin{theorem}\label{setstooddsets}
  For every $1 \leq k \leq t-2$, if there are $s_k$ $k$-sets in a minimal $t$-pseudo-border $\S$ then there are at least $\frac{n - (k-1)k}{k+1} s_k$ odd $k+1$-sets with respect to $\S$.
\end{theorem}

We give a proof of Theorem \ref{setstooddsets} in Section \ref{sect:oddsets}.

\begin{remark}
 Let $t \geq 2$. All 1-sets are odd in a minimal $t$-pseudo-border.
\end{remark}

Since there are exactly $n$ 1-sets, this remark immediately gives us

\begin{corollary}
Let $t \geq 3$. There are at least $\frac{n}{2}$ 2-sets in a minimal $t$-pseudo-border.
\end{corollary}

We can then apply Theorem \ref{setstooddsets} with $k=2$ to get

\begin{corollary}
  Let $t \geq 4$. There are at least $\frac{n}{2}\frac{n-2}{3} = \frac{n(n-2)}{2\cdot 3}$ odd 3-sets in a minimal $t$-pseudo-border.

  Let $t \geq 5$. There are at least $ \frac{n(n-2)}{2\cdot 3 \cdot 4}$ 4-sets in a minimal $t$-pseudo-border.
\end{corollary}

\subsection{Lower bounds for $t$-pseudo-borders}\label{sect:generallb}

Combining Theorem \ref{setstooddsets} with Lemma \ref{oddsettosets}, we obtain a lower bound on the number of $k$-sets in any minimal $t$-pseudo-border when $k \le \sqrt{n/2}$.  This concludes the proof of Theorem~\ref{theo:bound_pseudo-border}.

\begin{lemma}\label{cumulative-bound}
  For any minimal $t$-pseudo-border $\S$ and any even $k \le \min\{t-1, \sqrt{n/2}\}$, $\S$ has at least $\frac{n^{k/2}}{2^{k/2}k!}$ $k$-sets.
\end{lemma}

\begin{proof}
  Since $k \le \sqrt{n/2}$, $n - (k-1)k \ge \frac{n}{2}$.

  We prove this by induction on $k$. Since $\emptyset \in \S$, it is true for $k=0$. Suppose this is true for $k-2$. Then by Theorem \ref{setstooddsets}, there are at least $\frac{n}{2(k-1)}\frac{n^{(k-2)/2}}{2^{(k-2)/2}(k-2)!}$ odd $k+1$-sets with respect to $\S$. By Lemma \ref{oddsettosets}, $\S$ contains at least
  \[
  \frac{1}{k}\frac{n}{2(k-1)}\frac{n^{(k-2)/2}}{2^{(k-2)/2}(k-2)!}
  =\frac{n^{k/2}}{2^{k/2}k!}
  \]
  $k$-sets, as required.
\end{proof}

The bound in Lemma \ref{cumulative-bound} is maximized at $k=\sqrt{n/2}$ and by Stirling's formula is at least
\[
e^{\sqrt{n/2}}
\]

\comment{
  \begin{eqnarray*}
  \frac{n}{2} \prod_{k=1, k\text{ even}}^{\sqrt{n/2}} \frac{n - (k-1)k}{k+1}\frac{1}{k+2}
  &=& \frac{n}{2} \prod_{k=1, k\text{ even}}^{\sqrt{n/2}} \frac{n/2}{(k+1)(k+2)}\\
  &=& \frac{n}{2} \prod_{i=1}^{\sqrt{n/8}} \frac{n/2}{(2i+1)(2i+2)}\\
  &\ge& \frac{n^{\sqrt{n/8}}}{2^{\sqrt{n/8}}\sqrt{n/2}!}\\
  &\ge& \frac{n^{\sqrt{n/8}}}{2^{\sqrt{n/8}}(\sqrt{n}/e\sqrt{2})^\sqrt{n/2}}\\
  &=& \left(\frac{n}{2}\frac{2e^2}{n}\right)^{\sqrt{n/8}}\\
  &=& e^{2\sqrt{n/8}} 
  = e^{\sqrt{n/2}}
  \end{eqnarray*}
}

To obtain the lower bound on the size of $t$-pseudo-borders of $[n]$ stated in Theorem~\ref{theo:bound_pseudo-border}, we simply apply Lemma~\ref{cumulative-bound} to all the $k$-sets of a $t$-pseudo-border with $k<t$ and $k \leq \sqrt{n/2}$.

\subsection{Lower bounds for odd sets}\label{sect:oddsets}

This section is devoted to the proof of Theorem \ref{setstooddsets}.

%
%
%

\begin{proof}[Proof of Theorem \ref{setstooddsets}]
Let $\S$ be a minimal $t$-pseudo-border. We denote by $E$ the set of $k$-sets of $\S$ and by $F$ the odd $k+1$-sets with respect to $\S$.
For an element $v \in [n]$, we write $F_v$ for the sets of $F$ containing $v$, each with $v$ removed (so $F_v$ is a set of $k$-sets) and $E_v$ for the set of $k-1$-sets consisting of all elements of $E$ containing $v$ with $v$ itself removed from each $k$-set. 

Let $v$ be the vertex minimizing
  \[
  |F_v| + (k-1) |E_v| \le \frac{1}{n} \sum_v (|F_v| + (k-1) |E_v|)
  = \frac{1}{n} \sum_v |F_v| + \frac{k-1}{n} \sum_v |E_v|
  = \frac{k+1}{n}|F| + \frac{(k-1)k}{n} |E|
  \]

Since $\S$ is minimal, we may flip on $E_v$ and apply Lemma~\ref{lemma:flip-min}.

  We claim this flip yields a family $\S'$ whose $k$-sets is exactly $F_v$. Indeed, if $f\in F_v$ then $f \cup \{v\} \in F$ and so $f \cup \{v\}$ contains an odd number of elements of $E$. But except for $f$ itself, these elements of $E$ all contain $v$. If $f \not \in E$, $f$ contains an odd number of elements of $E_v$ and is therefore added (flipped) to $\S'$. If $f \in E$, $f$ contains an even number of elements of $E_v$ and is therefore not removed (flipped) when building $\S'$.

  On the other hand, if $f \not \in F_v$ then $f \cup \{v\} \not \in F$ which means $f \cup \{v\}$ contains an even number of elements of $E$. But except for $f$ itself, these elements of $E$ all contain $v$. If $f \not \in E$, $f$ contains an even number of elements of $E_v$ and is therefore not added (flipped) to $\S'$. If $f \in E$, $f$ contains an odd number of elements of $E_v$ and is therefore removed (flipped) when building $\S'$.

  The only other sets affected by flipping on $E_v$ are $k-2$ sets and this flips (gains) at most $(k-1) |E_v|$ elements of size $k-2$ (since each element of $E_v$ has size $k-1$ and contains at most $k-1$ elements (of $\S$) of size $k-2$).

  Therefore, by minimality (Lemma~\ref{lemma:flip-min}) $|E| \le |F_v| + (k-1) |E_v| \le \frac{k+1}{n}|F| + \frac{(k-1)k}{n} |E|$.
  Rearranging gives
  \begin{eqnarray*}
  \frac{n - (k-1)k}{n}|E| &\le& \frac{k+1}{n}|F|\\
  (n - (k-1)k)|E| &\le& (k+1)|F|\\
  \frac{n - (k-1)k}{k+1}|E| &\le&|F|
  \end{eqnarray*}
  and the theorem follows.
\end{proof}

\section*{Acknowledgements}

Nicolas Delfosse was supported by the Lockheed Martin Corporation. Nicolas Delfosse acknowledges the hospitality of Robert Raussendorf and the University of British Columbia where part of this article was written.


\begin{thebibliography}{10}

\bibitem{BM06:color}
H.~Bombin and M.A. Martin-Delgado.
\newblock Topological quantum distillation.
\newblock {\em Physical Review Letters}, 97:180501, 2006.

\bibitem{CS96}
A.R. Calderbank and P.W. Shor.
\newblock Good quantum error-correcting codes exist.
\newblock {\em Physical Review A}, 54(2):1098, 1996.

\bibitem{CDZ13}
A.~Couvreur, N.~Delfosse, and G.~Z\'emor.
\newblock A construction of quantum {LDPC} codes from {C}ayley graphs.
\newblock {\em Information Theory, IEEE Transactions on}, 59(9):6087--6098,
  2013.

\bibitem{De13:tradeoffs}
N.~Delfosse.
\newblock Tradeoffs for reliable quantum information storage in surface codes
  and color codes.
\newblock In {\em Proc. of IEEE International Symposium on Information Theory,
  ISIT 2013}, pages 917--921, 2013.

\bibitem{DZ13}
N.~Delfosse and G.~Z{\'e}mor.
\newblock Upper bounds on the rate of low density stabilizer codes for the
  quantum erasure channel.
\newblock {\em Quantum Information \& Computation}, 13(9-10):793--826, 2013.

\bibitem{Fe12}
E.~Fetaya.
\newblock Bounding the distance of quantum surface codes.
\newblock {\em Journal of Mathematical Physics}, 53:062202, 2012.

\bibitem{FML02}
M.H. Freedman, D.A. Meyer, and F.~Luo.
\newblock Z2-systolic freedom and quantum codes.
\newblock {\em Mathematics of Quantum Computation, Chapman \& Hall/CRC}, pages
  287--320, 2002.

\bibitem{Ga63}
R.~Gallager.
\newblock {\em Low Density Parity-Check Codes}.
\newblock PhD thesis, Massachusetts Institute of Technology, 1963.

\bibitem{Go13}
Daniel Gottesman.
\newblock What is the overhead required for fault-tolerant quantum computation?
\newblock {\em arXiv preprint arXiv:1310.2984}, 2013.

\bibitem{GL13}
Larry Guth and Alexander Lubotzky.
\newblock Quantum error-correcting codes and 4-dimensional arithmetic
  hyperbolic manifolds.
\newblock {\em arXiv preprint arXiv:1310.5555}, 2013.

\bibitem{Ki97}
A.Y. Kitaev.
\newblock Fault-tolerant quantum computation by anyons.
\newblock {\em Annals of Physics}, 303(1):27, 2003.

\bibitem{linial2006homological}
N.~Linial and R.~Meshulam.
\newblock Homological connectivity of random 2-complexes.
\newblock {\em Combinatorica}, 26(4):475--487, 2006.

\bibitem{MMS07}
D.~MacKay, G.~Mitchison, and A.~Shokrollahi.
\newblock More sparse-graph codes for quantum error-correction.
\newblock www.inference.phy.cam.ac.uk/mackay/cayley.pdf, 2007.

\bibitem{MMM04}
D.~J.~C. MacKay, G.~Mitchison, and P.~L. McFadden.
\newblock Sparse-graph codes for quantum error correction.
\newblock {\em IEEE Transaction on Information Theory}, 50(10):2315--2330,
  2004.

\bibitem{matousek}
J.~Matousek and U.~Wagner.
\newblock On {G}romov's method of selecting heavily covered points.
\newblock {\em arXiv preprint arXiv:1102.3515}, 2011.

\bibitem{meshulam2009homological}
R.~Meshulam and N.~Wallach.
\newblock Homological connectivity of random k-dimensional complexes.
\newblock {\em Random Structures \& Algorithms}, 34(3):408--417, 2009.

\bibitem{NC00}
M.A. Nielsen and I.L. Chuang.
\newblock {\em Quantum Computation and Quantum Information}.
\newblock Cambridge University Press, 1 edition, 2000.

\bibitem{St96}
A.~Steane.
\newblock Multiple-particle interference and quantum error correction.
\newblock {\em Proc. of the Royal Society of London. Series A: Mathematical,
  Physical and Engineering Sciences}, 452(1954):2551--2577, 1996.

\bibitem{TZ09}
J.-P. Tillich and G.~Z{\'e}mor.
\newblock Quantum {LDPC} codes with positive rate and minimum distance
  proportional to $n^{1/2}$;.
\newblock In {\em Proc. of IEEE International Symposium on Information Theory,
  ISIT 2009}, pages 799--803, 2009.

\bibitem{Ze09}
G.~Z{\'e}mor.
\newblock On {C}ayley graphs, surface codes, and the limits of homological
  coding for quantum error correction.
\newblock In {\em Proc. of the 2nd International Workshop on Coding and
  Cryptology, IWCC 2009}, pages 259--273. Springer-Verlag, 2009.

\end{thebibliography}

\end{document}